\documentclass[11pt,a4paper,fleqn,xcolor=dvipsnames]{article}

\makeatletter
\@ifclassloaded{llncs}{
}{}
\@ifclassloaded{beamer}{
    \beamertemplatenavigationsymbolsempty
    \definecolor{hks33}{RGB}{155,10,125}
    \definecolor{logobg}{RGB}{217,217,217}
    \setbeamerfont{title}{series=\bfseries}
    \setbeamercolor{title}{fg=hks33}
    \setbeamercolor{frametitle}{fg=black!70!white,bg=logobg}
}{}
\makeatother

\usepackage{graphicx}
\usepackage{tikz}
\usepackage{xcolor}
\usepackage[toc]{appendix}
\usepackage{setspace}
\usepackage{float}
\usepackage{amsfonts}
\usepackage{amssymb}
\usepackage{amsthm}
\usepackage{mathtools}
\usepackage{bm}
\usepackage{braket}
\usepackage[shortlabels]{enumitem}
\usepackage{algorithm}
\usepackage[noend]{algpseudocode}
\usepackage[utf8]{inputenc}
\usepackage{csquotes}
\usepackage{thmtools}
\usepackage{thm-restate}
\usepackage{nameref}
\usepackage{varioref}
\usepackage{hyperref}
\usepackage{cleveref}

\usetikzlibrary{decorations.pathreplacing,patterns,arrows.meta}

\makeatletter
\newcommand\useplaintitle{
    \def\@maketitle{
    	\newpage
    	\null
    	\vskip 2em%
    	\begin{center}%
    		\let \footnote \thanks
    		{\LARGE \@title \par}%
    		\vskip 1.5em%
    		{\large
    			\lineskip .5em%
    			\begin{tabular}[t]{c}%
    				\@author
    			\end{tabular}\par}%
    	\end{center}%
    	\par
    	\vskip 1em
    }
}
\makeatother

\DeclareTextFontCommand{\textbf}{\boldmath\bfseries}

\crefname{appsec}{Appendix}{Appendices}

\makeatletter\newcommand*{\currentname}{\@currentlabelname}\makeatother

\renewcommand{\emptyset}{\text{\O}}

\DeclarePairedDelimiter\ceil{\lceil}{\rceil}
\DeclarePairedDelimiter\floor{\lfloor}{\rfloor}
\DeclarePairedDelimiter\norm{\lVert}{\rVert}
\DeclarePairedDelimiter\abs{\lvert}{\rvert}

\newcommand{\op}[1]{\operatorname{#1}}
\newcommand{\oh}{\mathcal{O}}
\newcommand{\eps}{\varepsilon}
\newcommand{\integers}{\mathbb{Z}}
\newcommand{\cupdot}{\mathbin{\dot{\cup}}}
\newcommand{\NP}{\textsc{\textup{NP}}}

\DeclareMathOperator{\lcm}{lcm}

\usepackage{geometry}
\usepackage{comment}
\useplaintitle

\usepackage{fontawesome}
\usepackage{mdframed}

\title{Fuzzy Simultaneous Congruences}

\author{
	Max A. Deppert\footnote{Research supported by German Research Foundation (DFG) project JA 612/20-1}
	\\\small Kiel University
	\\\small Kiel, Germany
	\\\small made@informatik.uni-kiel.de
	\and
	Klaus Jansen\textsuperscript{\thefootnote}
	\\\small Kiel University
	\\\small Kiel, Germany
	\\\small kj@informatik.uni-kiel.de
	\and
	Kim-Manuel Klein
	\\\small Kiel University
	\\\small Kiel, Germany
	\\\small kmk@informatik.uni-kiel.de
}

\usepackage{titlesec}

\titleformat{\section}[runin]{\normalfont\bfseries}{\thesection}{0.5em}{}[.]
\titleformat{\subsection}[runin]{\normalfont\bfseries}{\thesubsection}{0.5em}{}[.]
\titleformat{\subsubsection}[runin]{\normalfont\bfseries}{\thesubsubsection}{0.5em}{}[.]
\titleformat{\paragraph}[runin]{\normalfont\bfseries}{\theparagraph}{0.5em}{}

\newtheorem{theorem}{Theorem}
\newtheorem{lemma}[theorem]{Lemma}
\newtheorem{observation}[theorem]{Observation}
\newtheorem{corollary}[theorem]{Corollary}

\Crefname{observation}{Observation}{Observations}

\newmdenv[leftmargin=1em,rightmargin=1em,innerleftmargin=0.4em,innerrightmargin=0.4em,innerbottommargin=1em,skipabove=1em,skipbelow=1em]{rectangle}

\providecommand{\keywords}[1]
{
  \small	
  \textbf{\textit{Keywords---}} #1
}

\begin{document}

\maketitle

\begin{abstract}
    We introduce a very natural generalization of the well-known problem of simultaneous congruences. Instead of searching for a positive integer $s$ that is specified by $n$ \emph{fixed remainders} modulo integer divisors $a_1,\dots,a_n$ we consider \emph{remainder intervals} $R_1,\dots,R_n$ such that $s$ is feasible if and only if $s$ is congruent to $r_i$ modulo $a_i$ for \emph{some} remainder $r_i$ in interval $R_i$ for all $i$.
	    
    This problem is a special case of a 2-stage integer program with only two variables per constraint which is is closely related to directed Diophantine approximation as well as the mixing set problem. We give a hardness result showing that the problem is NP-hard in general.
    
    By investigating the case of \emph{harmonic divisors}, i.e. $a_{i+1}/a_i$ is an integer for all $i<n$, which was heavily studied for the mixing set problem as well, we also answer a recent algorithmic question from the field of real-time systems. We present an algorithm to decide the feasibility of an instance in time $\oh(n^2)$ and we show that if it exists even the \emph{smallest} feasible solution can be computed in strongly polynomial time $\oh(n^3)$.
\end{abstract}
\keywords{Simultaneous congruences, Integer programming, Mixing Set, Real-time scheduling, Diophantine approximation}


\section{Introduction}

In the recent past there was a great interest in the so-called \emph{n-fold} IPs \cite{arXiv/2002.07745,DBLP:journals/mp/HemmeckeOR13,DBLP:conf/icalp/JansenLR19} and \emph{2-stage} IPs \cite{DBLP:journals/mp/HemmeckeS03,DBLP:conf/ipco/Klein20}. The matrix $\mathcal{A}$ of a 2-stage IP is constructed by blocks $A^{(1)}, \ldots, A^{(n)} \in \mathbb{Z}^{r \times k}$ and $B^{(1)}, \ldots, B^{(n)} \in \mathbb{Z}^{r \times t}$ as follows:
\begingroup\setlength\arraycolsep{2pt}
\[\mathcal{A}=\left(\begin{array}{ccccc}{A^{(1)}} & {B^{(1)}} & {0} & {\cdots} & {0} \\ {A^{(2)}} & {0} & {B^{(2)}} & {\ddots} & {\vdots} \\ {\vdots} & {\vdots} & {\ddots} & {\ddots} & {0} \\ {A^{(n)}} & {0} & {\cdots} & {0} & {B^{(n)}}\end{array}\right)\]
\endgroup
For an objective vector $c \in \integers_{\geq 0}^{k+nt}$, a right-hand side $b \in \integers^{nr}$, and bounds $\ell,u \in \integers_{\geq 0}^{k+nt}$ the 2-stage IP is formulated as
\begin{equation*}
\max \set{c^Tx|\mathcal{A} x=b, \;\ell \leq x \leq u,\; x \in \mathbb{Z}^{k+n t}}.
\end{equation*}
A special case of a 2-stage IP is given by the problem \textsc{Mixing Set} \cite{Conforti:2014:IP:2765770,DBLP:journals/siamdm/ConfortiSW07,DBLP:journals/mp/GunlukP01} (with only two variables in each constraint) where especially ${r=k=t=1}$ and ${A^{(1)} = \dots = A^{(n)}}$. Remark that 2-variable integer programming problems were extensively studied by various authors, e.g. \cite{DBLP:journals/algorithmica/Bar-YehudaR01,DBLP:conf/ipco/EisenbrandR01,DBLP:journals/siamcomp/Lagarias85}. The mixing set problem plays an important role for example in integer programming approaches for production planning \cite{Pochet:2006:PPM:1202598}. Given vectors $a,b\in\mathbb{Q}^n$ one aims to compute
\begin{equation}\label{mixing-set}
\min\set{f(s,x)|s+a_ix_i\geq b_i \forall i=1,\dots,n, (s,x)\in\mathbb{Z}_{\geq 0}\times\mathbb{Z}^n}
\end{equation}
for some objective function $f$. Conforti et al. \cite{DBLP:conf/ipco/ConfortiSW08} pose the question whether the problem can be solved in polynomial time for linear functions $f$. Unless P = NP this was ruled out by Eisenbrand and Rothvoß \cite{DBLP:conf/approx/EisenbrandR09} who proved that optimizing any linear function over a mixing set is NP-hard. However, the problem can be solved in polynomial time if $a_i = 1$ \cite{DBLP:journals/mp/GunlukP01,DBLP:journals/mp/MillerW03} or if the capacities $a_i$ fulfil a \emph{harmonic} property \cite{DBLP:journals/mp/ZhaoF08}, i.e. $a_{i+1}/a_i$ is integer for all $i < n$. The case of harmonic capacities was intensively studied - see \cite{DBLP:conf/ipco/ConfortiSW08,DBLP:journals/orl/ConfortiZ09} for simpler approaches.

	
More recently, real-time systems with harmonic tasks (the periods are pairs of integer multiples of each other) have received increased attention \cite{DBLP:conf/rtss/BonifaciMMW13} and also harmonic periods have been considered before \cite{DBLP:journals/jsa/AnssiKGT13,DBLP:conf/esa/EisenbrandKMNNSVW10,DBLP:conf/rtas/ShihGGCS03,d7c4f1a1-dce5-4ef9-a888-3089d808cec5}.
Now a recent manuscript in the field of real-time systems by Nguyen et al. \cite{arXiv/1912.01161} gives rise to the study of a new problem. They present an algorithm for the worst-case response time analysis of harmonic tasks with constrained release jitter running in polynomial time. The release jitter of a task is the maximum difference between the arrival times and the release times over all jobs of the task. Their algorithm uses heuristic components to solve an integer program that can be stated as a bounded version of the mixing set problem with additional upper bounds $B_i$ as follows.
	
\begin{rectangle}
	\textsc{Bounded Mixing Set (BMS)}\\
	Given capacities $a_1,\dots,a_n \in \mathbb{Z}$ and bounds $b,B \in \mathbb{Z}^n$ find $(s,x) \in \mathbb{Z}_{\geq 0} \times \mathbb{Z}^n$ such that
	\[b_i \leq s+a_i x_i \leq B_i \quad \forall i=1,\dots,n.\]
\end{rectangle}
    
In particular they depend on minimizing the value of $s$ which can be achieved in linear time in case of \textsc{Mixing Set}. See \Cref{mixing-set-minimize-s-is-simple} in the appendix for the short proof.
While BMS may look artificial at first sight it is not; in fact, leading to a very natural generalization it can be restated in the well-known form of \emph{simultaneous congruences}.
\begin{rectangle}
    \textsc{Fuzzy Simultaneous Congruences (FSC)}\\
    Given divisors $a_1,\dots,a_n \in \mathbb{Z}\setminus\!\{0\}$ and remainder intervals $R_1,\dots,R_n\subseteq \integers$\\ and an interval $S \subseteq \integers_{\geq 0}$ find a number $s \in S$ such that
    \[\exists \,r_i\in R_i: \,s \equiv r_i \!\!\!\pmod {a_i} \quad \forall i=1,\dots,n.\]
\end{rectangle}
Obviously, this also generalizes over the well-known problem of the Chinese Remainder Theorem. Here we give its generalized form (cf. \cite{DBLP:books/aw/Knuth81}).
\begin{theorem}[Generalized Chinese Remainder Theorem]
    Given divisors $a_1,\dots,a_n \in \mathbb{Z}_{\geq 1}$ and remainders $r_1,\dots,r_n \in \integers_{\geq 0}$ the system of $n$ simultaneous congruences
        \(s \equiv r_i \!\pmod {a_i}\)
    admits a solution $s \in \integers$ if and only if \(r_i\equiv r_j \!\pmod {\gcd(a_i,a_j)}\) for all $i\neq j$.
\end{theorem}
Furthermore, Leung and Whitehead \cite{DBLP:journals/pe/LeungW82} showed that $k$-Simultaneous Congruences ($k$-SC) is \NP-complete in the weak sense.
Given divisors $a_1,\dots,a_n \in \mathbb{Z}_{\geq 1}$ and remainders $r_1,\dots,r_n \in \integers_{\geq 0}$ the task is to find a number $s \in \integers_{\geq 0}$ and a subset $I \subseteq \{1,\dots,n\}$ with $|I|=k$ s.t.
    \(s \equiv r_i \!\!\!\pmod {a_i}\) for all $i\in I$.
Later it was shown by Baruah et al. \cite{DBLP:journals/rts/BaruahRH90} that \textsc{$k$-SC} also is \NP-complete in the strong sense.

Both problems BMS and FSC are interchangeable formulations of the same problem (see \Cref{sec:notation-and-properties}). Therefore, we will use them as synonyms and we especially assume formally that $R_i = [b_i,B_i]$.
Interestingly and to the best of our knowledge, FSC/BMS was not considered before. However, the investigation of simultaneous congruences has always been of transdisciplinary interest connecting a variety of fields and applications, e.g. \cite{DBLP:conf/focs/AgrawalB99,DBLP:conf/stoc/GoldreichRS99,DBLP:conf/focs/GuruswamiSS00}. 




\paragraph*{Our Contribution}
\begin{enumerate}[(a)]
    \item
    We show that BMS is \NP-hard for general capacities $a_i$. For the proof we refer to the appendix. Compared to the mixing set problem this is a stronger hardness result as BMS by itself only asks for an \emph{arbitrary} feasible solution. However, every mixing set problem may be solved by $s=\norm{b}_{\infty}, x=\mathbf{0}$.
    \item
    In the case of harmonic capacities (i.e. $a_{i+1}/a_i$ is an integer for all $i < n$), which was heavily studied for the mixing set problem as mentioned before, we give an algorithm exploiting a merge idea based on modular arithmetics on intervals to decide the feasibility problem of FSC in time $\oh(n^2)$. See \Cref{harmonic-feasibility} for the details.
    \item
    Furthermore, for a feasible instance of FSC with harmonic capacities we present a polynomial algorithm as well as a strongly polynomial algorithm to compute 
    the smallest feasible solution to FSC in time $\oh(\min\{n^2\log(a_n),n^3\}) \leq \oh(n^3)$. See \Cref{harmonic-optimization} for the details.
    \item
    Our algorithm gives a strongly polynomial replacement for the heuristic component (which may fail to compute a solution) in the algorithm of Nguyen et al. \cite{arXiv/1912.01161}. However, we present an algorithm to solve the problem in linear time. See \Cref{real-time-scheduling} for the details.
\end{enumerate}

\section{Notation and General Properties}
\label{sec:notation-and-properties}

For the sake of readability we write $X^{[\alpha]} = (X \bmod \alpha)$ for numbers $X$ as well as $X^{[\alpha]} = \set{z \bmod \alpha|z\in X}$ for sets $X$ (of numbers) to denote the modular projection of some number or interval, respectively.
Extending the usual notation we also write $X \equiv Y \pmod \alpha$ if $X^{[\alpha]} = Y^{[\alpha]}$ for sets $X, Y$. Notice that on the one hand $(X \cup Y)^{[\alpha]} = X^{[\alpha]} \cup Y^{[\alpha]}$ but on the other hand be aware that $(X \cap Y)^{[\alpha]} \neq X^{[\alpha]} \cap Y^{[\alpha]}$ in general (cf. \Cref{modular-intersections}). \Cref{fig:modular-projection} depicts the structure of $v^{[\alpha]}$ if $v=[\ell_v,u_v]$ is an interval in $\integers$.

Also we use the well-tried notation $t + X = \set{t+z|z\in X}$ to express the \emph{translation} of a set of numbers $X$ by some number $t$. For a set of sets $\mathcal{S}$ we write $\bigcup \mathcal{S}$ to denote the union $\bigcup_{S \in \mathcal{S}} S$. Furthermore, we identify constraints by their indices. So for $i \leq n$ we say that
\enquote{$b_i \leq s+a_ix_i \leq B_i$} \emph{is} constraint $i$.

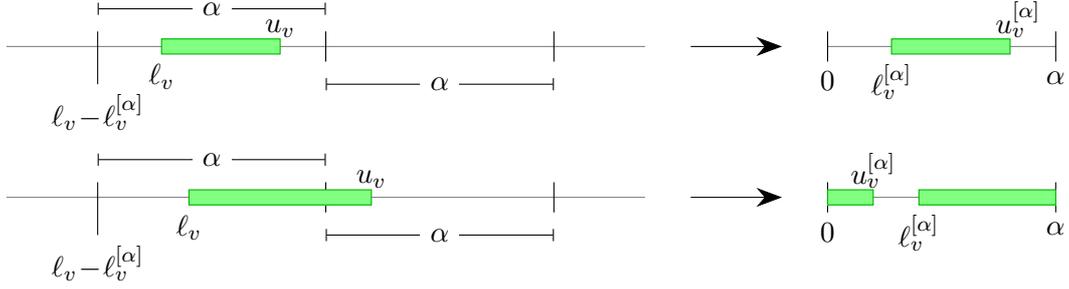
\begin{figure}
    \centering
	\begin{tikzpicture}[xscale=1.2]
        \usetikzlibrary{arrows.meta}
\draw [gray] (-1,2) 
-- (6,2);
\draw (0,1.5) node [below] {$\ell_v \!-\! \ell_v^{[\alpha]}$} -- (0,2.2);
\draw (2.5,1.8) node [below] {} -- (2.5,2.2);
\draw (5,1.8) node [below] {} -- (5,2.2);
\draw [|-|] (0,2.5) -- node[fill=white] {$\alpha$} (2.5,2.5);
\draw [|-|] (2.5,1.5) -- node[fill=white] {$\alpha$} (5,1.5);
\fill [green!50!white,draw=green!70!black] (0.7,1.9) node [below,black] {$\ell_v$} rectangle (2,2.1) node [above=-3pt,black] {$u_v$};
\draw[-{Stealth[scale=2]}] (6.5,2) -- (7.5,2);
\draw [gray] (8,2) -- (10.5,2);
\draw (8,1.8) node [below] {$0$} -- (8,2.2);
\draw (10.5,1.8) node [below] {$\alpha$} -- (10.5,2.2);
\fill [green!50!white,draw=green!70!black] (8.7,1.9) node [below,black] {$\ell_v^{[\alpha]}$} rectangle (10,2.1) node [right=3pt,above=-3pt,black] {$u_v^{[\alpha]}$};
\draw [gray] (-1,0) 
-- (6,0);
\draw (0,-0.5) node [below] {$\ell_v \!-\! \ell_v^{[\alpha]}$} -- (0,0.2);
\draw (2.5,-0.2) node [below] {} -- (2.5,0.2);
\draw (5,-0.2) node [below] {} -- (5,0.2);
\draw [|-|] (0,0.5) -- node[fill=white] {$\alpha$} (2.5,0.5);
\draw [|-|] (2.5,-0.5) -- node[fill=white] {$\alpha$} (5,-0.5);
\fill [green!50!white,draw=green!70!black] (1,-0.1) node [below,black] {$\ell_v$} rectangle (3,0.1) node [above=-3pt,black] {$u_v$};
\draw[-{Stealth[scale=2]}] (6.5,0) -- (7.5,0);
\draw [gray] (8,0) -- (10.5,0);
\draw (8,-0.2) node [below] {$0$} -- (8,0.2);
\draw (10.5,-0.2) node [below] {$\alpha$} -- (10.5,0.2);
\fill [green!50!white,draw=green!70!black] (8,-0.1) rectangle (8.5,0.1) node [above=-3pt,black] {$u_v^{[\alpha]}$};
\fill [green!50!white,draw=green!70!black] (9,-0.1) node [below,black] {$\ell_v^{[\alpha]}$} rectangle (10.5,0.1);
	\end{tikzpicture}
    \caption{The two possibilities for the modular projection of an interval}
    \label{fig:modular-projection}
\end{figure}

\paragraph*{Identity of BMS and FSC}

In fact, BMS allows zero capacities while FSC cannot allow zero divisors since $(\mathrm{mod}\ 0)$ is undefined. However, consider a constraint $i$ of BMS with $a_i \neq 0$.
Let $b_i \leq s +a_ix_i \leq B_i$ be satisfied and set $r_i = s +a_i x_i$. Then $r_i^{[a_i]} = s^{[a_i]}$ and $r_i \in [b_i,B_i]=R_i$.
Vice-versa let $r_i \in R_i$ s.t. $r_i \equiv s \pmod {a_i}$. Then there is an $x_i \in \mathbb{Z}$ s.t. $s+a_ix_i = r_i \in R_i = [b_i,B_i]$.

A constraint $i$ that holds $a_i = 0$ simply requires that $s \in R_i$. Hence, if $a_i = a_j = 0$ for two constraints $i\neq j$ they can be replaced by one new constraint $k$ defined by $R_k = R_i \cap R_j$. Therefore, one may assume that there is at most one constraint $i$ with a zero capacity $a_i$.
As all our results can be lifted to the general case with low effort we will assume in terms of BMS that all capacities are non-zero and for FSC we take the equivalent assumption that $S=\integers_{\geq 0}$.\\

With our notation we may easily express the feasibility of a value $s$ for a single constraint $i$ as follows.
\begin{observation}\label{s-feasible-for-one-constraint}
    A value $s$ satisfies constraint $i$ if and only if $s^{[a_i]} \in R_i^{[a_i]}$.
\end{observation}
\begin{proof}
    $\exists r_i \in R_i: r_i \equiv s \pmod {a_i}$
    iff $\exists r_i \in R_i: r_i^{[a_i]} = s^{[a_i]}$ iff
    $s^{[a_i]} \in R_i^{[a_i]}$.
\end{proof}

By simply swapping the signs of the $x_i$ we may assume that $a_i \geq 0$ for all $i$.
We may also assume that the intervals are small in the sense that $B_i - b_i + 1 < a_i$ holds for all $i$. Assume that $B_i - b_i + 1\geq a_i$ for an $i$ and let $s \geq 0$ be an arbitrary integer. Then $b_i \leq B_i - a_i + 1$ and constraint $i$ may always be solved by setting $x_i = \ceil{(b_i-s)/a_i}$ which satisfies
\[b_i
\leq s + a_i\underbrace{\ceil{\tfrac{b_i-s}{a_i}}}_{x_i}
\leq s + a_i\ceil{\tfrac{B_i-a_i+1-s}{a_i}}
= s + a_i\floor{\tfrac{B_i-s}{a_i}}
\leq B_i.\]
Hence, constraint $i$ is redundant and may be omitted. As a direct consequence there can be at most one feasible value for each $x_i$ for a given guess $s$.
In fact, 
we can decide the feasibility of a guess $s$ in time $\oh(n)$ as for all constraints $i$ and values $x_i$ it holds
$b_i \leq s + a_i x_i \leq B_i$ if and only if $\ceil{(b_i-s)/a_i} = x_i = \floor{(B_i-s)/a_i}$.
So a guess $s$ is feasible if and only if $\ceil{(b_i-s)/a_i} = \floor{(B_i-s)/a_i}$ holds for all constraints $i$. By $s_{\min}$ we denote the smallest feasible solution $s$ that satisfies all constraints.

\begin{observation}\label{s-min-at-most-lcm}
	For feasible instances it holds that $\op{s_{\min}} < \lcm(a_1,\dots,a_n)$.
\end{observation}
\begin{proof}
	Let $\varphi = \lcm(a_1,\dots,a_n)$. Remark that $\varphi/a_i$ is integral for all $i$. Assume that $(s,x)$ is a solution with $s=s_{\min} \geq \varphi$. Let $t = s-\varphi$ and $y_i = x_i+\varphi/a_i$ f.a. $i$. Then $0 \leq t < s_{\min}$ and
	\(t+a_iy_i = s+a_ix_i\) f.a. $i$. So $(t,y)$ is a solution that contradicts the optimality.
\end{proof}

\section{Harmonic Divisors}

Here we consider harmonic divisors in the sense that $a_{i+1}/a_i$ is an integer for all $i < n$. We present an algorithm to decide the feasibility of an instance of FSC. Also we show how optimal solutions can be computed in (strongly) polynomial time. Both of these results are based on the fine-grained interconnection between modular arithmetics on sets and the harmonic property.
For some intuition \Cref{fig:lattice-idea-example} gives a perspective on $s$ as an anchor for $1$-dimensional lattices with basis $a_i$ which have to \enquote{hit} the intervals $R_i$. For example, in the figure it holds that $s + a_2\cdot(-1) = s - a_2 \in R_2$, so the 1-dimensional lattice $(s+a_2z)_{z \in \integers}$ hits interval $R_2$. Therefore, the choice of $s$ satisfies constraint $2$.

\subsection{Deciding feasibility}
\label{harmonic-feasibility}

The idea for our first algorithm will be to decide the feasibility problem by iteratively computing modular projections from constraint $i=n$ down to $i=1$.
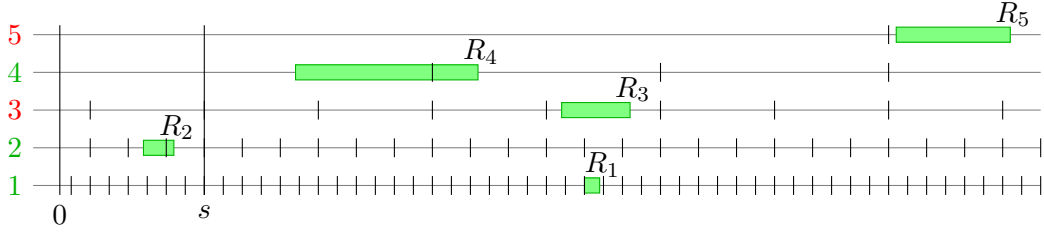
\begin{figure}
    \centering
    \begin{tikzpicture}
        \def\x{1}
        \draw [gray] (-0.25*\x,0.5) node[left=0.3,green!70!black] {$1$}  -- (\x*13,0.5);
\draw [gray] (-0.25*\x,1.0) node[left=0.3,green!70!black] {$2$} -- (\x*13,1.0);
\draw [gray] (-0.25*\x,1.5) node[left=0.3,red] {$3$} -- (\x*13,1.5);
\draw [gray] (-0.25*\x,2.0) node[left=0.3,green!70!black] {$4$} -- (\x*13,2.0);
\draw [gray] (-0.25*\x,2.5) node[left=0.3,red] {$5$}-- (\x*13,2.5);
\fill [green!50!white,draw=green!70!black] (7*\x,0.5-0.1) rectangle (7.2*\x,0.5+0.1) node [black,xshift=1,yshift=5] {$R_1$};
\fill [green!50!white,draw=green!70!black] (1.2*\x,1-0.1) rectangle (1.6*\x,1+0.1) node [black,xshift=1,yshift=5] {$R_2$};
\fill [green!50!white,draw=green!70!black] (6.7*\x,1.5-0.1) rectangle (7.6*\x,1.5+0.1) node [black,xshift=1,yshift=5] {$R_3$};
\fill [green!50!white,draw=green!70!black] (3.2*\x,2-0.1) rectangle (5.6*\x,2+0.1) node [black,xshift=1,yshift=5] {$R_4$};
\fill [green!50!white,draw=green!70!black] (11.1*\x,2.5-0.1) rectangle (12.6*\x,2.5+0.1) node [black,xshift=1,yshift=5] {$R_5$};
\foreach \z in {0.25,0.5,...,13} \draw (\x*\z,0.375) -- (\x*\z,0.625);
\foreach \z in {0.5,1.0,...,13} \draw (\x*\z,0.5+0.375) -- (\x*\z,0.5+0.625);
\foreach \z in {0.5,2.0,...,13} \draw (\x*\z,1.0+0.375) -- (\x*\z,1.0+0.625);
\foreach \z in {5.0,8.0,...,13} \draw (\x*\z,1.5+0.375) -- (\x*\z,1.5+0.625);
\foreach \z in {11.0} \draw (\x*\z,2+0.375) -- (\x*\z,2+0.625);
\draw (0.1*\x,0.375) node[below] {$0$} -- (0.1*\x,2.625);
\draw (2*\x,0.375) node[below] {$s$} -- (2*\x,2.625);
    \end{tikzpicture}
    \caption{$36a_1 \!=\! 18a_2 \!=\! 6a_3 \!=\! 3a_4 \!=\! a_5$. The guess $s$ is not feasible for constr. $3$ and $5$}
    \label{fig:lattice-idea-example}
\end{figure}
In the following we will say that an interval $w \subseteq \integers$ \emph{represents} of a set $M \subseteq \integers$ (modulo $\alpha$) if $w^{[\alpha]} = M^{[\alpha]}$. Also a set of intervals $\mathcal{R}$ represents a set $M \subseteq \integers$ (modulo $\alpha$) if $M^{[\alpha]} = \bigcup_{w\in \mathcal{R}} w^{[\alpha]}$. 
Given an integer $\alpha \geq 1$ and two intervals $v$, $w$ we depend on the structure of the intersection $v^{[\alpha]} \cap w^{[\alpha]} \subseteq [0,\alpha)$.
To express it let $v = [\ell_v,u_v], w = [\ell_w,u_w]$ and we define the basic intervals
\[
    \varphi_{\alpha}(v,w) = [\ell_v^{[\alpha]},u_w^{[\alpha]}]
    \quad\text{and}\quad
    \psi_{\alpha}(v,w) = [\max\{\ell_v^{[\alpha]},\ell_w^{[\alpha]}\},\alpha+\min\{u_v^{[\alpha]},u_w^{[\alpha]}\}]
\]
for all intervals $v$, $w$. Remark that $\psi_{\alpha}(w,v) = \psi_{\alpha}(v,w)$ is always true.

\begin{lemma}\label{alpha-intersection-possibilities}
    Given an integer $\alpha \geq 1$ and two intervals $v, w \subseteq \integers$ it holds that
    \begin{align*}
        v^{[\alpha]} &\cap w^{[\alpha]} \in \{\;\;
         \emptyset,\;\;
        v^{[\alpha]},\;\;
        w^{[\alpha]},\;\;
        \psi_{\alpha}(v,w)^{[\alpha]},\;\;
        \varphi_{\alpha}(v,w),\;\;
        \varphi_{\alpha}(w,v),\\
        & \varphi_{\alpha}(v,w) \cupdot \varphi_{\alpha}(w,v),\;\;
        \varphi_{\alpha}(v,w) \cupdot \psi_{\alpha}(v,w)^{[\alpha]},\;\;
        \varphi_{\alpha}(w,v) \cupdot \psi_{\alpha}(v,w)^{[\alpha]}
        \;\;\}.
    \end{align*}
\end{lemma}

\begin{figure}
    \centering
    \begin{tikzpicture}[xscale=.11,yscale=.1]
        \def\a{28}\def\b{25}
        \begin{scope}[shift={(0,0)}]
  \draw (0,-8) node [below] {$0$} -- (0,8);
  \draw (20,-8) node [below] {$\alpha$} -- (20,8);
  \draw (0,5-1) rectangle (8,5+1); \draw (14,5-1) rectangle (20,5+1);
  \draw (9,-1) rectangle (12,1);
  \node [draw] at (10,-10) {0};

  \node [left] at (0,5) {$w^{[\alpha]}$:};
  \node [left] at (0,0) {$v^{[\alpha]}$:};
  \node [left] at (0,-5) {$v^{[\alpha]} \cap w^{[\alpha]}$:};
\end{scope}
\begin{scope}[shift={(1*\a,0)}]
  \draw (0,-8) node [below] {$0$} -- (0,8);
  \draw (20,-8) node [below] {$\alpha$} -- (20,8);
  \draw (0,5-1) rectangle (8,5+1); \draw (14,5-1) rectangle (20,5+1);
  \draw (0,-1) rectangle (6,1);\draw (10,-1) rectangle (20,1);
  \draw (0,-5-1) rectangle (6,-5+1);\draw (14,-5-1) rectangle (20,-5+1);
  \node [draw] at (10,-10) {1.1};
\end{scope}
\begin{scope}[shift={(2*\a,0)}]
  \draw (0,-8) node [below] {$0$}  -- (0,8);
  \draw (20,-8) node [below] {$\alpha$}  -- (20,8);
  \draw (0,5-1) rectangle (13,5+1);\draw (16,5-1) rectangle (20,5+1);
  \draw (0,-1) rectangle (6,1);\draw (10,-1) rectangle (20,1);
  \draw (0,-5-1) rectangle (6,-5+1);\draw (10,-5-1) rectangle (13,-5+1);\draw (16,-5-1) rectangle (20,-5+1);
  \node [draw] at (10,-10) {1.2};
\end{scope}
\begin{scope}[shift={(3*\a,0)}]
  \draw (0,-8) node [below] {$0$}  -- (0,8);
  \draw (20,-8) node [below] {$\alpha$}  -- (20,8);
  \draw (0,5-1) rectangle (6,5+1);\draw (10,5-1) rectangle (20,5+1);
  \draw (0,-1) rectangle (13,1);\draw (16,-1) rectangle (20,1);
  \draw (0,-5-1) rectangle (6,-5+1);\draw (10,-5-1) rectangle (13,-5+1);\draw (16,-5-1) rectangle (20,-5+1);
  \node [draw] at (10,-10) {1.3};
\end{scope}
\begin{scope}[shift={(0,-1*\b)}]
  \draw (0,-8) node [below] {$0$}  -- (0,8);
  \draw (20,-8) node [below] {$\alpha$}  -- (20,8);
  \draw (0,5-1) rectangle (8,5+1);\draw (16,5-1) rectangle (20,5+1);
  \draw (5,-1) rectangle (14,1);
  \draw (5,-5-1) rectangle (8,-5+1);
  \node [draw] at (10,-10) {2.1};

  \node [left] at (0,5) {$w^{[\alpha]}$:};
  \node [left] at (0,0) {$v^{[\alpha]}$:};
  \node [left] at (0,-5) {$v^{[\alpha]} \cap w^{[\alpha]}$:};
\end{scope}
\begin{scope}[shift={(1*\a,-1*\b)}]
  \draw (0,-8) node [below] {$0$}  -- (0,8);
  \draw (20,-8) node [below] {$\alpha$}  -- (20,8);
  \draw (0,5-1) rectangle (8,5+1);\draw (16,5-1) rectangle (20,5+1);
  \draw (11,-1) rectangle (18,1);
  \draw (16,-5-1) rectangle (18,-5+1);
  \node [draw] at (10,-10) {2.2};
\end{scope}
\begin{scope}[shift={(2*\a,-1*\b)}]
  \draw (0,-8) node [below] {$0$}  -- (0,8);
  \draw (20,-8) node [below] {$\alpha$}  -- (20,8);
  \draw (0,5-1) rectangle (9,5+1);\draw (15,5-1) rectangle (20,5+1);
  \draw (6,-1) rectangle (17,1);
  \draw (6,-5-1) rectangle (9,-5+1);\draw (15,-5-1) rectangle (17,-5+1);
  \node [draw] at (10,-10) {2.3};
\end{scope}
\begin{scope}[shift={(3*\a,-1*\b)}]
  \draw (0,-8) node [below] {$0$}  -- (0,8);
  \draw (20,-8) node [below] {$\alpha$}  -- (20,8);
  \draw (3,5-1) rectangle (13,5+1);
  \draw (9,-1) rectangle (18,1);
  \draw (9,-5-1) rectangle (13,-5+1);
  \node [draw] at (10,-10) {3};
\end{scope}
    \end{tikzpicture}
    \caption{Examples for the cases of the case distinction in the proof of \Cref{alpha-intersection-possibilities}}
    \label{fig:intersection-distinction}
\end{figure}
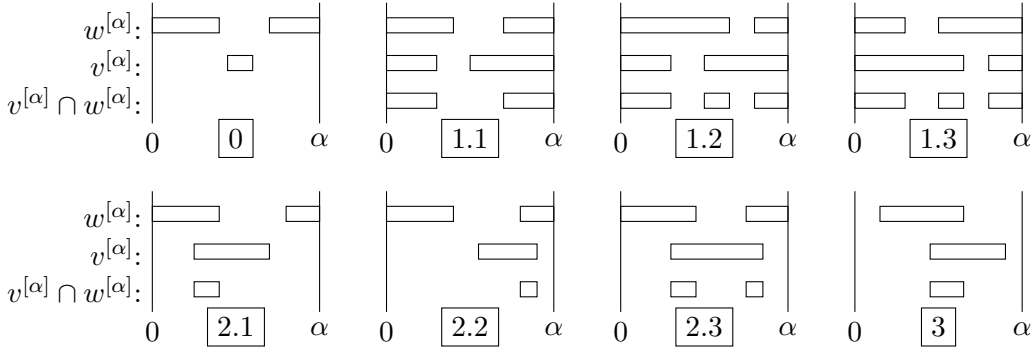

The important intuition is that such a \enquote{modulo $\alpha$ intersection} can always be represented by at most two intervals. Remark that the sets in the second row are the only ones which are represented by $2>1$ intervals.
\begin{proof}
    We do a case distinction (see \Cref{fig:intersection-distinction}) as follows.
    We only look at the non-trivial case, i.e. $v^{[\alpha]} \cap w^{[\alpha]} \notin \set{\emptyset,v^{[\alpha]},w^{[\alpha]}}$, which especially implies $\abs{v} < \alpha$ and $\abs{w} < \alpha$. 
    
    We start with the case that neither $v^{[\alpha]}$ nor $w^{[\alpha]}$ is an interval, i.e. $u_v^{[\alpha]} < \ell_v^{[\alpha]}$ and $u_w^{[\alpha]} < \ell_w^{[\alpha]}$. 
    Then it cannot be that $u_w^{[\alpha]} \geq \ell_v^{[\alpha]}$ and $u_v^{[\alpha]} \geq \ell_w^{[\alpha]}$ since that implies $\ell_v^{[\alpha]} \leq u_w^{[\alpha]} < \ell_w^{[\alpha]} \leq u_v^{[\alpha]}$. Hence, there are three cases as follows.
    
    \noindent\emph{Case 1.1}. $u_w^{[\alpha]} < \ell_v^{[\alpha]}$ and $u_v^{[\alpha]} < \ell_w^{[\alpha]}$. Then the intersection equals
    \begin{align*}
        [0,\min\{u_v^{[\alpha]},u_w^{[\alpha]}\}] \cupdot [\max\{\ell_v^{[\alpha]},\ell_w^{[\alpha]}\},\alpha)
        &= [\max\{\ell_v^{[\alpha]},\ell_w^{[\alpha]}\},\alpha + \min\{u_v^{[\alpha]},u_w^{[\alpha]}\}]^{[\alpha]}\\
        &= \psi_{\alpha}(v,w)^{[\alpha]}.
    \end{align*}
    
    \noindent\emph{Case 1.2}. $u_w^{[\alpha]} \geq \ell_v^{[\alpha]}$ and $u_v^{[\alpha]} < \ell_w^{[\alpha]}$. Then the intersection equals
    \[
        [0,u_v^{[\alpha]}] \cupdot [\ell_v^{[\alpha]},u_w^{[\alpha]}] \cupdot [\ell_w^{[\alpha]},\alpha)
        = [\ell_v^{[\alpha]},u_w^{[\alpha]}] \cupdot [\ell_w^{[\alpha]},\alpha+u_v^{[\alpha]}]^{[\alpha]}
        = \varphi_{\alpha}(v,w) \cupdot \psi_{\alpha}(v,w)^{[\alpha]}.
    \]
    
    \noindent\emph{Case 1.3}. $u_w^{[\alpha]} < \ell_v^{[\alpha]}$ and $u_v^{[\alpha]} \geq \ell_w^{[\alpha]}$.
    By symmetry we get $v^{[\alpha]} \cap w^{[\alpha]} = \varphi_{\alpha}(w,v) \cupdot \psi_{\alpha}(v,w)^{[\alpha]}$.

    Now, w.l.o.g. assume that $v^{[\alpha]}$ is an interval, i.e. $\ell_v^{[\alpha]} \leq u_v^{[\alpha]}$, while $w^{[\alpha]}$ consists of two intervals, i.e. $u_w^{[\alpha]} < \ell_w^{[\alpha]}$. Then there are three cases as follows.
    
    \noindent\emph{Case 2.1}. $\ell_v^{[\alpha]} \leq u_w^{[\alpha]} < u_v^{[\alpha]} < \ell_w^{[\alpha]}$.
    Then the intersection equals $[\ell_v^{[\alpha]},u_w^{[\alpha]}] = \varphi_{\alpha}(v,w)$.
    
    \noindent\emph{Case 2.2}. $u_w^{[\alpha]} < \ell_v^{[\alpha]} < \ell_w^{[\alpha]} \leq u_v^{[\alpha]}$.
    Then the intersection equals $[\ell_w^{[\alpha]},u_v^{[\alpha]}] = \varphi_{\alpha}(w,v)$.
    
    \noindent\emph{Case 2.3}. $\ell_v^{[\alpha]} \leq u_w^{[\alpha]} < \ell_w^{[\alpha]} \leq u_v^{[\alpha]}$.
    Then the intersection is
    \[[\ell_v^{[\alpha]},u_w^{[\alpha]}] \cupdot [\ell_w^{[\alpha]},u_v^{[\alpha]}] = \varphi_{\alpha}(v,w) \cupdot \varphi_{\alpha}(w,v).\]
    Clearly, if both $v^{[\alpha]}$ and $w^{[\alpha]}$ are intervals (\emph{Case 3}) (which are not disjoint) then their intersection is either $\varphi_{\alpha}(v,w)$ or $\varphi_{\alpha}(w,v)$.
\end{proof}

While the previous lemma characterized the form of intersections of two modular projections of intervals, the next lemma reveals how many intervals will be required to represent a \emph{one-to-many} intersection. We will use this bound in every step of our algorithm.
We want to add that both of these lemmata and even \Cref{feasibility-lemma} do not depend on the harmonic property by themselves. However, they turn out to be especially useful in this setting.

\begin{lemma}\label{k+1-many-intervals}
    Let $\alpha \geq 1$, let $v$ be an interval and let $Q$ be a set of $k\geq 1$ intervals.
    Then there is a set $R$ of at most $k\!+\!1$ intervals s.t.
    \(
        v^{[\alpha]} \cap (\bigcup Q)^{[\alpha]} = (\bigcup R)^{[\alpha]}
    \).
\end{lemma}
\begin{proof}[Proof of \Cref{k+1-many-intervals}]
    We simply obtain that
    \[
        v^{[\alpha]} \cap \left(\bigcup Q\right)^{[\alpha]}
        = \bigcup_{w\in Q} (v^{[\alpha]} \cap w^{[\alpha]})
        = \bigcup_{\substack{w\in Q\\w^{[\alpha]} \subseteq v^{[\alpha]}}} \!\!\!\!\!\!\! w^{[\alpha]} \,\cup \bigcup_{w\in D} (v^{[\alpha]} \cap w^{[\alpha]})
    \]
    where $D = \set{w\in Q|w^{[\alpha]} \nsubseteq v^{[\alpha]}, w^{[\alpha]} \cap v^{[\alpha]} \neq \emptyset}$ denotes the subset of intervals that cause the interesting intersections with $v^{[\alpha]}$ (cf. \Cref{alpha-intersection-possibilities}).
    Obviously, all other intersections can be represented by at most one interval each. So we study the intersections with $D$. In fact, everything gets simple if there are $w_1,w_2 \in D$ such that $v^{[\alpha]} \cap w_1^{[\alpha]} = \varphi_{\alpha}(v,w_1) \cupdot \psi_{\alpha}(v,w_1)^{[\alpha]}$ and $v^{[\alpha]} \cap w_2^{[\alpha]} = \varphi_{\alpha}(w_2,v) \cupdot \psi_{\alpha}(v,w_2)^{[\alpha]}$. By simply adapting the inequalities of the first case distinction in the proof of \Cref{alpha-intersection-possibilities} we find
    \begin{align*}
        (v^{[\alpha]} &\cap w_1^{[\alpha]}) \cup (v^{[\alpha]} \cap w_2^{[\alpha]})\\
        &= ([0,u_v^{[\alpha]}] \cupdot [\ell_v^{[\alpha]},u_{w_1}^{[\alpha]}] \cupdot [\ell_{w_1}^{[\alpha]},\alpha)) \cup ([0,u_{w_2}^{[\alpha]}] \cupdot [\ell_{w_2}^{[\alpha]},u_v^{[\alpha]}] \cupdot [\ell_v^{[\alpha]},\alpha))\\
        &= [0,u_v^{[\alpha]}] \cupdot [\ell_v^{[\alpha]},\alpha) = v^{[\alpha]}
    \end{align*}
    which implies that $v^{[\alpha]} \cap (\bigcup Q)^{[\alpha]} = v^{[\alpha]}$ can be represented by only one interval, namely $v$. Therefore, in order to get an upper bound we assume that these two types of intersections do not come together. In more detail, we may assume by symmetry that $D = D_1 \cupdot D_2$ where
    \begin{align*}
        & D_1 = \set{w\in D|v^{[\alpha]} \cap w^{[\alpha]} = \varphi_{\alpha}(v,w) \cupdot \varphi_{\alpha}(w,v)} \text{ and}\\
        & D_2 = \set{w\in D|v^{[\alpha]} \cap w^{[\alpha]} = \varphi_{\alpha}(v,w) \cupdot \psi_{\alpha}(v,w)^{[\alpha]}}.
    \end{align*}
    It turns out that
    \begin{align*}
        \bigcup_{w\in D_1} \! (v^{[\alpha]} \cap w^{[\alpha]})
        &= \!\!\bigcup_{w\in D_1} \! ([\ell_v^{[\alpha]},u_w^{[\alpha]}] \cupdot [\ell_w^{[\alpha]},u_v^{[\alpha]}])\\
        &= [\ell_v^{[\alpha]},\max_{w\in D_1}u_w^{[\alpha]}] \cup [\min_{w\in D_1}\ell_w^{[\alpha]},u_v^{[\alpha]}]\quad\text{and}\\
        \bigcup_{w\in D_2} \! (v^{[\alpha]} \cap w^{[\alpha]})
        &= \bigcup_{w\in D_2} ([\ell_v^{[\alpha]},u_w^{[\alpha]}] \cupdot [\ell_w^{[\alpha]},\alpha+u_v^{[\alpha]}]^{[\alpha]})\\
        &= [\ell_v^{[\alpha]},\max_{w\in D_2} u_w^{[\alpha]}] \cup [\min_{w\in D_2}\ell_w^{[\alpha]},\alpha+u_v^{[\alpha]}]^{[\alpha]}
    \end{align*}
    which finally joins up to
    \[\bigcup_{w\in D} \! (v^{[\alpha]} \cap w^{[\alpha]}) = [\ell_v^{[\alpha]},\max_{w\in D} u_w^{[\alpha]}] \cup [\min_{w\in D}\ell_w^{[\alpha]},\alpha+u_v^{[\alpha]}]^{[\alpha]}.\]
    Hence, all intersections with intervals in $D$ may be represented by at most two intervals in total while each other intersection can be represented by at most one interval. Thus, if $\abs{D}=0$ then the whole intersection can be represented by at most $k$ intervals. If $\abs{D}\geq 1$ then there are at most $2+\abs{Q}-\abs{D} \leq 2+k-1 = k+1$ intervals required.
\end{proof}
\begin{figure}
    \centering
    \begin{tikzpicture}
    	\def\x{.65}\def\y{.7}\draw (0,3.5*\y) node [left] {$Q_{i+1}$:} -- (20*\x,3.5*\y);
\draw (0,3.5*\y-.2) node [below] {$0$} -- (0,3.5*\y+.2);
\draw (4*\x,3.5*\y-.2) node [below] {$a_i$} -- (4*\x,3.5*\y+.2);
\draw (8*\x,3.5*\y-.2) -- (8*\x,3.5*\y+.2);
\draw (12*\x,3.5*\y-.2) -- (12*\x,3.5*\y+.2);
\draw (16*\x,3.5*\y-.2) -- (16*\x,3.5*\y+.2);
\draw (20*\x,3.5*\y-.2) node [below] {$a_{i+1}$} -- (20*\x,3.5*\y+.2);
\fill [green!50!white,draw=green!70!black] (4.5*\x,3.5*\y-.15) rectangle (6*\x,3.5*\y+.15);
\fill [green!50!white,draw=green!70!black] (14.5*\x,3.5*\y-.1) rectangle (17*\x,3.5*\y+.1);

\draw (0,2*\y)  -- (4*\x,2*\y);
\draw (0,2*\y-.2) node [below] {$0$} -- (0,2*\y+.2);
\draw (4*\x,2*\y-.2) node [below] {$a_i$} -- (4*\x,2*\y+.2);
\fill [green!50!white,draw=green!70!black] (0.5*\x,2*\y-.15) rectangle (2*\x,2*\y+.15);
\fill [green!50!white,draw=green!70!black] (2.5*\x,2*\y-.1) rectangle (4*\x,2*\y+.1);
\fill [green!50!white,draw=green!70!black] (0,2*\y-.1) rectangle (\x,2*\y+.1);

\fill [lightgray,draw=darkgray] (0.7*\x,1.3*\y-.1) rectangle node [below,black] {$R_i^{[a_i]}$} (3.3*\x,1.3*\y+.1);

\draw (0,0) node [left] {$Q_i$:} -- (4*\x,0);
\draw (0,-0.2) node [below] {$0$} -- (0,0.2);
\draw (4*\x,-0.2) node [below] {$a_i$} -- (4*\x,0.2);
\fill [green!50!white,draw=green!70!black] (0.7*\x,-0.1) rectangle (2*\x,0.1);
\fill [green!50!white,draw=green!70!black] (2.5*\x,-0.1) rectangle (3.3*\x,0.1);
    \end{tikzpicture}
    \caption{A step from $i\!+\!1$ to $i$; modular projection to $[0,a_i)$ and intersection with $R_i^{[a_i]}$}
    \label{fig:feasibility-test-step}
\end{figure}
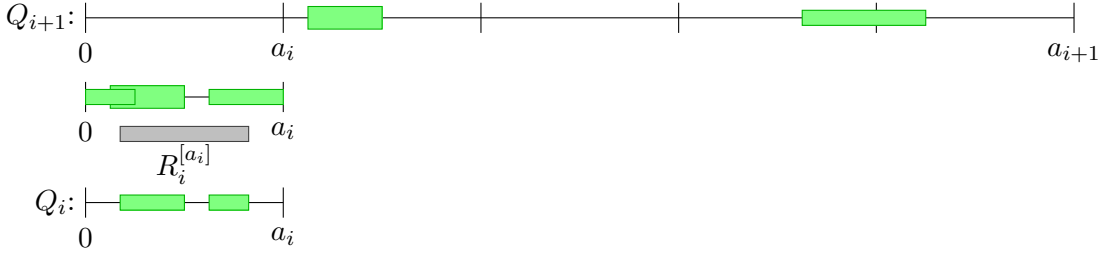
\begin{algorithm}
\begin{algorithmic}
\Procedure{Feasible}{$I=(a_1,\dots,a_n,R_1,\dots,R_n)$}
    \State $Q_n \gets \{R_n\}$
    \For{$i=n-1,\dots,1$}
        \State Compute set $Q_i$ s. t. $(\bigcup Q_i)^{[a_i]} = R_i^{[a_i]}\cap(\bigcup Q_{i+1})^{[a_i]}$ and $|Q_i| \leq \oh(n-i)$
    \EndFor
    \If{$\bigcup Q_1 = \emptyset$}
        \State
        \Return{\enquote{infeasible}}
    \Else
        \State
        \Return{\enquote{feasible}}
    \EndIf
\EndProcedure
\end{algorithmic}
\caption{Feasibility test for FSC}
\label{alg:feasible}
\end{algorithm}
Let $S_i$ denote the set of all solutions $s \in \integers_{\geq 0}$ that are feasible for each of the constraints $i,i+1,\dots,n$. We set $S_{n+1} = \mathbb{Z}_{\geq 0}$ to denote the feasible solutions to an empty set of constraints. The correctness of \Cref{alg:feasible} is implied by the following fundamental lemma. See \Cref{fig:feasibility-test-step} for an example of a step inside the algorithm.
\begin{lemma}\label{feasibility-lemma}
    It holds true that $S_i^{[a_i]} = R_i^{[a_i]} \cap S_{i+1}^{[a_i]}$ for all $i=1,\dots,n$.
\end{lemma}
\begin{proof}
    Let $r\in S_i^{[a_i]}$. So there is a solution $s \in S_i$ such that $r=s^{[a_i]} \in R_i^{[a_i]}$. It holds that $S_i \subseteq S_{i+1}$ which implies $s \in S_{i+1}$ and thus $r=s^{[a_i]} \in S_{i+1}^{[a_i]}$.
    
    Vice-versa let $r \in R_i^{[a_i]} \cap S_{i+1}^{[a_i]}$. So there is a solution $s\in S_{i+1}$ with $s^{[a_i]}=r$. From $r\in R_i^{[a_i]}$ we get $s^{[a_i]}\in R_i^{[a_i]}$. Hence, $s\in S_i$ and $r=s^{[a_i]} \in S_i^{[a_i]}$.
\end{proof}

\begin{theorem}
    \Cref{alg:feasible} decides the feasibility of an instance in time $\oh(n^2)$.
\end{theorem}
\begin{proof}
    We show that $\bigcup Q_i \equiv S_i \pmod {a_i}$ for all $i=n,\dots,1$. This will prove the algorithm correct since then $\bigcup Q_1 \equiv S_1 \pmod {a_1}$ and that means $\bigcup Q_1$ is empty if and only if $S_1$ is empty.
    Obviously it holds that $\bigcup Q_n \equiv S_n \pmod {a_n}$ since $\bigcup Q_n = R_n$.
    Now suppose that $\bigcup Q_{i+1} \equiv S_{i+1} \pmod {a_{i+1}}$ for some $i\geq 1$. We have that
    \[\left(\bigcup Q_i\right)^{[a_i]} = R_i^{[a_i]}\cap\left(\bigcup Q_{i+1}\right)^{[a_i]}\]
    where the harmonic property implies \[\left(\bigcup Q_{i+1}\right)^{[a_i]} = \left(\left(\bigcup Q_{i+1}\right)^{[a_{i+1}]}\right)^{[a_i]} = \left(S_{i+1}^{[a_{i+1}]}\right)^{[a_i]} = S_{i+1}^{[a_i]}.\] Together with \Cref{feasibility-lemma} this yields
    \[
        \left(\bigcup Q_i\right)^{[a_i]} = R_i^{[a_i]}\cap S_{i+1}^{[a_i]} = S_i^{[a_i]}
    \]
    and that proves the algorithm correct.
    Using \Cref{alpha-intersection-possibilities,k+1-many-intervals,feasibility-lemma} each set $Q_i$ can be computed in time $\oh(n)$ and this yields a total running time of $\oh(n^2)$.
\end{proof}

\subsection{Computing optimal solutions}
\label{harmonic-optimization}

Unfortunately, \Cref{alg:feasible} neither calculates a solution nor directly implies one. Here we show how to compute the smallest feasible solution $s_{\min}$ to FSC. However, by searching in the opposite direction the same technique also applies to the computation of the largest feasible solution $s_{\max} < a_n$. We start with a simple binary search approach.

\begin{corollary}\label{minimize-s-with-binary-search}
    For feasible instances $s_{\min}$ can be computed in time $\oh(n^2\log(a_n))$.
\end{corollary}
This can be achieved by introducing an additional constraint measuring the value of $s$ as follows.
Let $\beta$ be a positive integer. We extend the problem instance by a new constraint with number $n+1$ defined by
$a_{n+1} = 2\cdot a_n$, $b_{n+1} = 0$, and $B_{n+1} = \beta$.
Remark that this \emph{$\beta$-instance} admits the same set of solutions as the original instance as long as $\beta$ is large enough, e.g. $\beta = a_n$ (cf. \Cref{s-min-at-most-lcm}).
Consider a feasible solution to the $\beta$-instance where $\beta \leq a_n$. It holds that
\[2a_nx_{n+1} = a_{n+1}x_{n+1} \leq s + a_{n+1}x_{n+1} \leq B_{n+1} = \beta \leq a_n\]
which implies $x_{n+1} \leq \floor{\tfrac12} = 0$.
However, if $x_{n+1} < 0$ then $s \geq a_{n+1}\cdot \abs{x_{n+1}}$ and therefore the solution $s' = s + a_{n+1}x_{n+1}$ with $x_{n+1}' = 0$ and $x_i' = x_i - (a_{n+1}/a_i)x_{n+1}$ for all $i=1,\dots,n$ is better than $s$ and $x_{n+1}' = 0$.
    
Thus we may assume generally that $x_{n+1} = 0$ which allows us to measure the value of $s$ using the upper bound $\beta$.
We use $\beta$ to do a binary search in the interval $[0,a_n]$ using \Cref{alg:feasible} to check the $\beta$-instance for feasibility. The smallest possible value for $\beta$ then states the optimum value and that proves \Cref{minimize-s-with-binary-search}.
However, with additional ideas we are able to achieve strongly polynomial time.
The next lemma seems to be a characteristic property of modular arithmetics on sets.

\begin{lemma}\label{modular-intersections}
    For all numbers $a,b \in \integers_{\geq 1}$ and sets $A,B \subseteq \integers$ it holds
    \[A^{[a]}\cap B^{[a]} = \left(A^{[ab]} \cap \bigcup_{i=0}^{b-1}(ia+B^{[a]})\right)^{[a]}.\]
\end{lemma}
\begin{proof}
Let $x$ be a number. Then it holds
\begin{align*}
    x \in \left(A^{[ab]} \cap \bigcup_{i=0}^{b-1}(ia+B^{[a]})\right)^{[a]}
    &\Leftrightarrow\;\; \exists y \in A^{[ab]}: y \in \bigcup_{i=0}^{b-1}(ia+B^{[a]}) \land x = y^{[a]}\\
    &\Leftrightarrow\;\; \exists y \in A^{[ab]} : y^{[a]} \in B^{[a]} \land x = y^{[a]}\\
    &\Leftrightarrow\;\; x \in A^{[a]} \cap B^{[a]}
\end{align*}
where the last equivalence follows from $(A^{[ab]})^{[a]} = A^{[a]}$.
\end{proof}

Since the right side can be written as the modular projection of a \emph{union of intersections} modulo $a$ we can find a sensible strengthening; in fact, for arbitrary sets $X,M_0,\dots,M_{m-1}$ it holds that
\[\bigcup_{i=0}^{m-1}(X \cap M_i) = \bigcup_{i=0}^{m-1} (X \cap (M_i\setminus \bigcup_{j=0}^{i-1} (X \cap M_j))).\]
While the left-hand side may not, the right-hand side is always a \emph{disjoint} union. Taking into account the modular projections this leads to the following corollary.
\begin{corollary}\label{modular-intersections-disjoint}
    For all numbers $a,b \in \integers_{\geq 1}$ and sets $A,B \subseteq \integers$ it holds
    \[A^{[a]}\cap B^{[a]} = \left(\bigcup_{i=0}^{b-1} D_i\right)^{[a]}\]
    where $D_i = A^{[ab]} \cap Y_i$ and
    $Y_i = i a + (B^{[a]} \setminus \bigcup_{j=0}^{i-1} D_j^{[a]})$ for all $i=0,\dots,b-1$.
\end{corollary}

We will use \Cref{modular-intersections-disjoint} to aggregate constraints in order to reduce the problem size. The following observation gives a first bound for the smallest feasible solution.

\begin{observation}
	For feasible instances it holds that $s_{\min} \in R_n^{[a_n]}$.
\end{observation}
This is true since in the harmonic case $s_{\min} < \lcm(a_1,\dots,a_n) = a_n$ due to \Cref{s-min-at-most-lcm} which then implies that $s_{\min} = s_{\min}^{[a_n]} \in R_n^{[a_n]}$ using \Cref{s-feasible-for-one-constraint}.


\noindent The idea is to search for $s_{\min}$ in the modular projection $R_n^{[a_n]}$ by aggregating the penultimate constraint $n-1$ into the last constraint $n$. Fortunately, the number of intervals needed to represent both constraints can be bounded by a constant. A fine-grained construction then enforces the algorithm to efficiently iterate the feasibility test on aggregated instances to find the optimum value.

\begin{theorem}
    For feasible instances $s_{\min}$ can be computed in time $\oh(n^3)$.
\end{theorem}

Remark that the set of feasible solutions for the last two constraints is $S_{n-1} = R_{n-1}^{[a_{n-1}]} \cap (R_n^{[a_n]})^{[a_{n-1}]} = R_{n-1}^{[a_{n-1}]} \cap R_n^{[a_{n-1}]}$. Therefore, the next lemma states the crucial argument of the algorithm.
\begin{figure}
    \centering
    \begin{tikzpicture}
        \def\x{.82}\def\y{.5}\usetikzlibrary{patterns,decorations.pathreplacing}

\draw (0,4*\y) node [left] {$R_n^{[a_n]}$:} -- (16*\x,4*\y);

\draw (4*\x,4*\y-0.2) node [below] {$a_{n-1}$} -- (4*\x,4*\y+0.2);
\draw (8*\x,4*\y-0.2) -- (8*\x,4*\y+0.2);
\draw (12*\x,4*\y-0.2) -- (12*\x,4*\y+0.2);
\fill [green!50!white,draw=green!70!black] (0*\x,4*\y-.1) rectangle (0.8*\x,4*\y+.1);
\fill [green!50!white,draw=green!70!black] (9.7*\x,4*\y-.1) rectangle (16*\x,4*\y+.1);
\draw (0,4*\y-0.2) node [below] {$0$} -- (0,4*\y+0.2);
\draw (16*\x,4*\y-0.2) node [below] {$a_n$}  -- (16*\x,4*\y+0.2);

\draw (0,2*\y) -- (16*\x,2*\y);

\draw [dashed] (0*\x,2*\y-.1) -- (0*\x,5*\y-.3);
\draw [dashed] (2.4*\x,2*\y-.1) -- node [rotate=90] {} (2.4*\x,5*\y-.3);
\fill [lightgray,draw=darkgray] (0*\x,2*\y-.1) rectangle (2.4*\x,2*\y+.1);

\draw [dashed] (3.2*\x,2*\y-.1) -- node [rotate=90] {} (3.2*\x,5*\y-.3);
\draw [dashed] (4*\x,2*\y-.1) -- node [rotate=90] {} (4*\x,5*\y-.3);

\fill [lightgray,draw=darkgray] (3.2*\x,2*\y-.1) rectangle (4*\x,2*\y+.1);

\draw [white, pattern=north east lines, pattern color=red] (8*\x,2*\y-.1) rectangle (8.8*\x,2*\y+.1);

\draw [dashed] (0.8*\x+8*\x,2*\y-.1) -- node [rotate=90] {} (0.8*\x+8*\x,5*\y-.3);
\draw [dashed] (2.4*\x+8*\x,2*\y-.1) -- node [rotate=90] {} (2.4*\x+8*\x,5*\y-.3);

\fill [lightgray,draw=darkgray] (0.8*\x+8*\x,2*\y-.1) rectangle (2.4*\x+8*\x,2*\y+.1);

\draw [dashed] (3.2*\x+8*\x,2*\y-.1) -- node [rotate=90] {} (3.2*\x+8*\x,5*\y-.3);
\draw [dashed] (4*\x+8*\x,2*\y-.1) -- node [rotate=90] {} (4*\x+8*\x,5*\y-.3);

\draw [white, pattern=north east lines, pattern color=red] (12*\x,2*\y-.1) rectangle (12.8*\x,2*\y+.1);
\fill [lightgray,draw=darkgray] (3.2*\x+8*\x,2*\y-.1) rectangle (4*\x+8*\x,2*\y+.1);

\draw [dashed] (0.8*\x+12*\x,2*\y-.1) -- node [rotate=90] {} (0.8*\x+12*\x,5*\y-.3);
\draw [dashed] (1.7*\x+12*\x,2*\y-.1) -- node [rotate=90] {} (1.7*\x+12*\x,5*\y-.3);

\draw [white, pattern=north east lines, pattern color=red] (13.7*\x,2*\y-.1) rectangle (14.4*\x,2*\y+.1);
\fill [lightgray,draw=darkgray] (0.8*\x+12*\x,2*\y-.1) rectangle (1.7*\x+12*\x,2*\y+.1);

\draw [white, pattern=north east lines, pattern color=red] (3.2*\x+12*\x,2*\y-.1) rectangle (4*\x+12*\x,2*\y+.1);
    
\draw [decorate,decoration={brace,amplitude=10*\y,mirror}] (0*\x,1.7*\y) -- node[below,xshift=28*\x] {$Y_0\!=\!R_{n-1}^{[a_{n-1}]}$} (4*\x,1.7*\y);
\draw [decorate,decoration={brace,amplitude=10*\y,mirror}] (8.8*\x,1.7*\y) -- node[below=7*\y] {$Y_2$} (12*\x,1.7*\y);
\draw [decorate,decoration={brace,amplitude=10*\y,mirror}] (0.8*\x+12*\x,1.7*\y) -- node[below=7*\y] {$Y_3$} (1.7*\x+12*\x,1.7*\y);
    
\draw (0,0) -- (16*\x,0);

\fill [green!50!white,draw=green!70!black] (0*\x,-.1) rectangle (0.8*\x,.1);
\fill [green!50!white,draw=green!70!black] (9.7*\x,-.1) rectangle (10.4*\x,.1);
\fill [green!50!white,draw=green!70!black] (3.2*\x+8*\x,-.1) rectangle (12*\x,.1);
\fill [green!50!white,draw=green!70!black] (0.8*\x+12*\x,-.1) rectangle (1.7*\x+12*\x,.1);
\draw (4*\x,-0.2) node [below] {$a_{n-1}$} -- (4*\x,0.2);
\draw (8*\x,-0.2) -- (8*\x,0.2);
\draw (12*\x,-0.2) -- (12*\x,0.2);
\draw (0,-0.2) node [below] {$0$} -- (0,0.2);
\draw (16*\x,-0.2) node [below] {$a_n$} -- (16*\x,.2);
    
\draw [decorate,decoration={brace,amplitude=10*\y,mirror}] (0*\x,-0.5*\y) -- node[below=7*\y,xshift=4*\x] {$D_0$} (0.8*\x,-0.5*\y);
\draw [decorate,decoration={brace,amplitude=10*\y,mirror}] (9.7*\x,-0.5*\y) -- node[below=7*\y] {$D_2$} (12*\x,-0.5*\y);
\draw [decorate,decoration={brace,amplitude=10*\y,mirror}] (0.8*\x+12*\x,-0.5*\y) -- node[below=7*\y] {$D_3$} (1.7*\x+12*\x,-0.5*\y);
    \end{tikzpicture}
    \caption{An example of four required intervals to represent \(R_{n-1}^{[a_{n-1}]} \cap R_n^{[a_{n-1}]}\) in \Cref{at-most-four-smallest-intervals}}
    \label{fig:four-intervals}
\end{figure}
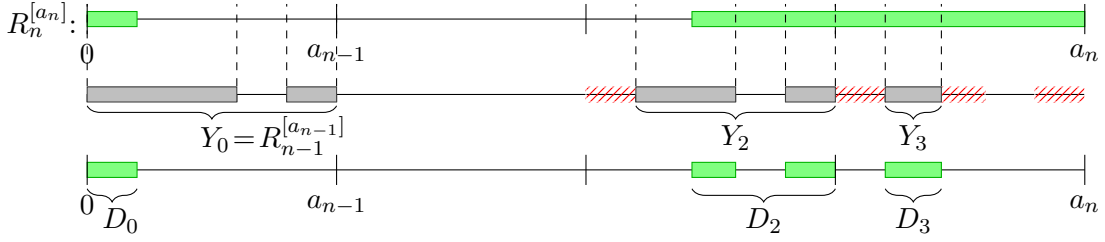
\begin{lemma}\label{at-most-four-smallest-intervals}
    The intersection $R_{n-1}^{[a_{n-1}]} \cap R_n^{[a_{n-1}]}$ can always be represented by the disjoint union $U \subseteq R_n^{[a_n]}$ of only constant many intervals in $R_n^{[a_n]}$ such that
    \begin{enumerate}[(a)]
        \item \(U^{[a_{n-1}]} = R_{n-1}^{[a_{n-1}]} \cap R_n^{[a_{n-1}]}\) and
        \item $u \equiv r \pmod {a_{n-1}}$ implies $u \leq r$ for all $u\in U$, $r \in R_n^{[a_n]}$.
    \end{enumerate}
\end{lemma}
Here the former property states that indeed the intervals in $U$ are a proper representation for the last two constraints. The important property is the latter; in fact, it ensures that $U$ is the best possible representation in the sense that $U$ consists of the \emph{smallest} intervals possible (see \Cref{fig:four-intervals}).

\begin{proof}[Proof of \Cref{at-most-four-smallest-intervals}]
    (a). By defining $D_i = Y_i \cap R_n^{[a_n]}$ and
\[Y_i = i a_{n-1} + (R_{n-1}^{[a_{n-1}]} \setminus \bigcup_{j=0}^{i-1} D_j^{[a_{n-1}]})\]
    for all $i \in \{0,\dots,a_n/a_{n-1}-1\}$ \Cref{modular-intersections-disjoint} proves the claim (cf. \Cref{fig:four-intervals}).
    (b) follows by construction.

    It remains to show that $\bigcup_i D_i$ is the union of only constant many disjoint intervals. Apparently, the intervals are disjoint by construction.
    
    We claim that there are at most three non-empty sets $D_i$. Assume there are at least four non-empty translates $D_i$, namely $D_i,D_j,D_k,D_{\ell}$. Then, since $R_n$ is an interval it holds for at least two $p,q \in \{i,j,k,\ell\}$ that the \emph{full interval translates} $F_p = [pa_{n-1},(p+1)a_{n-1})$ and $F_q = [qa_{n-1},(q+1)a_{n-1})$ are subsets of $R_n^{[a_n]}$. For $p$ (and also for $q$) we get
    
    \[D_p^{[a_{n-1}]} = ({\underbrace{Y_p}_{\subseteq F_p}} \cap R_n^{[a_n]})^{[a_{n-1}]} = Y_p^{[a_{n-1}]} = R_{n-1}^{[a_{n-1}]} \setminus \bigcup_{j=0}^{p-1}D_j^{[a_{n-1}]}\]
    which implies with $\bigcup_{j=0}^{p-1} D_j^{[a_{n-1}]} \subseteq R_{n-1}^{[a_{n-1}]}$ that
    \[\bigcup_{j=0}^p D_j^{[a_{n-1}]} = D_p^{[a_{n-1}]} \cup \bigcup_{j=0}^{p-1} D_j^{[a_{n-1}]} = R_{n-1}^{[a_{n-1}]}.\]
    Then it follows
    \(\bigcup_{j=0}^p D_j^{[a_{n-1}]} = R_{n-1}^{[a_{n-1}]} = \bigcup_{j=0}^q D_j^{[a_{n-1}]}\).
    W.l.o.g. let $p < q$. Then $D_q = Y_q \cap R_n^{[a_n]}$ is empty since
    \[Y_q
    = qa_{n-1} + \left(R_{n-1}^{[a_{n-1}]} \setminus \bigcup_{j=0}^{q-1} D_j^{[a_{n-1}]}\right)
    \subseteq qa_{n-1} + \left(R_{n-1}^{[a_{n-1}]} \setminus R_{n-1}^{[a_{n-1}]}\right)
    \]
    is empty and we have a contradiction.
    
    Using the same case distinctions as in the proof of \Cref{alpha-intersection-possibilities} one can show that each set $D_i$ consist of at most two intervals. Therefore, all the non-empty sets $D_i$ consist of at most $3\cdot 2 = 6$ intervals in total. In fact, one can improve this bound to a total number of at most $4$ intervals (see \Cref{fig:four-intervals}) by a more sophisticated case distinction.
\end{proof}

This admits an algorithm using an aggregation argument as follows. For constraints $n$ and $n-1$ we use \Cref{at-most-four-smallest-intervals} to compute disjoint intervals $E_1,\dots,E_k \subseteq R_n^{[a_n]}$ (representing the constraints $n$ and $n-1$) where $k \leq C$ for a small constant $C$. If $k\geq 1$ then use \Cref{alg:feasible} to check the feasibility of the instances $\mathcal{I}_1,\dots,\mathcal{I}_k$ defined by
\[(\mathcal{I}_j) \qquad \min \set{s|s^{[a_i]} \in R_i^{[a_i]} \forall i=1,\dots,n-2,\;s^{[a_n]} \in E_j^{[a_n]},s \in \mathbb{Z}_{\geq 0}}.\]
If none of the instances $\mathcal{I}_1,\dots,\mathcal{I}_k$ admits a solution then the original instance can not be feasible. Assume that there is at least one feasible instance. Now, since $E_1,\dots,E_k$ are disjoint exactly one of them contains the optimum value for $s$.  W.l.o.g. assume that $E_1 < \dots < E_k$. Then there is a smallest index $j$ such that $\mathcal{I}_j$ is feasible and we solve $\mathcal{I}_j$ recursively to find the optimum value.
Together this yields an algorithm running in time $n\cdot C \cdot\oh(n^2) = \oh(n^3)$.

\section{Uniprocessor Real-Time Scheduling}
\label{real-time-scheduling}

In real-time systems an important question is to ask for the \emph{worst-case response time} of a system. Nguyen et al. proposed a new algorithm \cite{arXiv/1912.01161} to compute it in polynomial time for preemptive sporadic tasks $\tau_1,\dots,\tau_n$ with harmonic periods $T_i \geq 0$ running on a uniprocessor platform. Be aware that they assume the harmonic property in the opposite direction, i.e. $T_i/T_{i+1} \in \mathbb{Z}$. Their algorithm even allows the task execution to be delayed by some release jitter $J_i$. However, their algorithm depends on a heuristic component which may fail to find a solution. In fact, the fundamental computation problem can be expressed as a BMS instance which immediately implies a robust solution in time $\mathcal{O}(n^3)$ with our algorithm. Nevertheless, it can be solved even more efficiently in time $\mathcal{O}(n)$ which we describe here.

We adapt the notation of Nguyen et al. and extend it to our needs. The jobs of task $\tau_i$ have the processing time $C_i$ and we define $c_i = \sum_{t=i+1}^{n-1}C_t$ to accumulate the last of them. The utilization of task $\tau_i$ is denoted by $U_i = C_i/T_i$ and it holds that $\sum_{t=1}^{n-1} U_t < 1$. In Section 5.4.1 of \cite{arXiv/1912.01161} Nguyen et al. describe that also $x_1 = 1$ may be assumed. The system to solve is
\begin{equation}\label{response-system}
    \begin{split}
        \min\,\{\,x_n \mid \; &J_i + T_ix_i \leq J_n + T_nx_n, \\
        & J_n+T_nx_n-c_i \leq J_i+T_ix_i \;\; \forall i\leq n-1\,\}
    \end{split}
\end{equation}
which can be formulated as the following BMS instance:
\begin{equation}\label{response-bms}
    \min\Set{x_n|\ceil*{\frac{J_i-J_n}{T_n}} \leq x_n - \frac{T_i}{T_n}x_i \leq \floor*{\frac{J_i-J_n+c_i}{T_n}} \;\; \forall i\leq n-1}
\end{equation}

\begin{lemma}\label{uniqueness}
    If $i < j \leq n$ and $(c_i+c_j)/T_j < 1$ then in terms of variable $x_i$ there is at most one feasible value for variable $x_j$.
\end{lemma}
\begin{proof}
    If $j < n$ then by combining the constraints for $i$ and $j$ in \eqref{response-system} we find
    \begin{align*}
        T_ix_i + J_i - J_n \; &\leq \; T_jx_j + J_j-J_n+c_j\quad\text{ and}\\
        T_jx_j+J_j-J_n \; &\leq \; T_ix_i + J_i - J_n + c_i
    \end{align*}
    which with the harmonic property and the integrality of $x_j$ yields
    \begin{equation}\label{next-variable-bounds}
        \frac{T_i}{T_j}x_i+\ceil*{\frac{J_i-J_j-c_j}{T_j}} \leq x_j \leq \frac{T_i}{T_j}x_i+\floor*{\frac{J_i-J_j+c_i}{T_j}}.
    \end{equation}
    However, if $j=n$ then $c_j = \sum_{t=n+1}^{n-1}C_t = 0$ and thus \eqref{next-variable-bounds} follows from \eqref{response-system} too (cf. \eqref{response-bms}). Now by simply dropping the roundings we obtain in both cases that
    \[
        \frac{T_i}{T_j}x_i+\floor*{\frac{J_i-J_j+c_i}{T_j}} - \left(\frac{T_i}{T_j}x_i+\ceil*{\frac{J_i-J_j-c_j}{T_j}}\right)
        \leq \frac{c_i+c_j}{T_j} < 1
    \]
    which proves the claim.
\end{proof}

According to \eqref{next-variable-bounds} we define interval bounds $\ell_j^{(i)}(z)$ and $u_j^{(i)}(z)$ to denote the feasible values for variable $x_j$ in terms of variable $x_i$ where $z$ states a value for variable $x_i$, i.e. \[\ell_j^{(i)}(z) = \frac{T_i}{T_j}z+\ceil*{\frac{J_i-J_j-c_j}{T_j}} \qquad \text{and} \qquad u_j^{(i)}(z) = \frac{T_i}{T_j}z+\floor*{\frac{J_i-J_j+c_i}{T_j}}.\]
Thus, \eqref{next-variable-bounds} is equivalent to $x_j \in [\ell_j^{(i)}(x_i),u_j^{(i)}(x_i)]$ and if $(c_i+c_j)/T_j < 1$ then it either holds that $\ell_j^{(i)}(x_i) = x_j = u_j^{(i)}(x_i)$ or there is no solution at all.

Fortunately, there is always a chain of variables such that the value of every next variable can be determined by knowing the value of the previous. The following lemma is crucial.

\begin{figure}
    \centering
    \begin{tikzpicture}[scale=0.859]
        \usetikzlibrary{arrows.meta}

\fill [lightgray,draw=black] (0,0) rectangle node [color=black,minimum height=1cm] (x_k0) {$x_1$} (1,1);
\fill [lightgray,draw=black] (3,0) rectangle node [minimum height=1cm] (x_k1) {} (4,1);
\fill [lightgray,draw=black] (5,0) rectangle node [color=black,minimum height=1cm] (x_k2) {$x_i$} (6,1);
\draw (6,0) rectangle node [minimum height=1cm] {$x_{i+1}$} (7,1);
\fill [lightgray,draw=black] (11,0) rectangle node [color=black,minimum height=1cm] (x_k3) {$x_k$} (12,1);
\fill [lightgray,draw=black] (13,0) rectangle node [minimum height=1cm] (x_k4) {} (14,1);
\fill [lightgray,draw=black] (14,0) rectangle node [minimum height=1cm] (x_k5) {} (15,1);
\fill [lightgray,draw=black] (16,0) rectangle node [color=black,minimum height=1cm] (x_k6) {$x_n$} (17,1);

\draw (0,0) rectangle (17,1);

\draw[very thick] (1,-0.5) -- (1,1.5);
\draw[very thick] (4,-0.5) -- (4,1.5);
\draw[very thick] (6,-0.5) -- (6,1.5);
\draw[very thick] (12,-0.5) -- (12,1.5);
\draw[very thick] (14,-0.5) -- (14,1.5);
\draw[very thick] (15,-0.5) -- (15,1.5);

\draw [-{Stealth[scale=2]},rounded corners=5pt] (x_k0) -- + (0,1.5) -| (x_k1) -- + (0,-1.5) -| (x_k2) -- + (0,+1.5) -| (x_k3) -- + (0,-1.5) -| (x_k4) -- + (0,+1.5) -| (x_k5) -- + (0,-1.5) -| (x_k6);
    \end{tikzpicture}
    \caption{The variable revealing flow with vertical lines between blocks of equal periods}
    \label{fig:revealing-flow}
\end{figure}
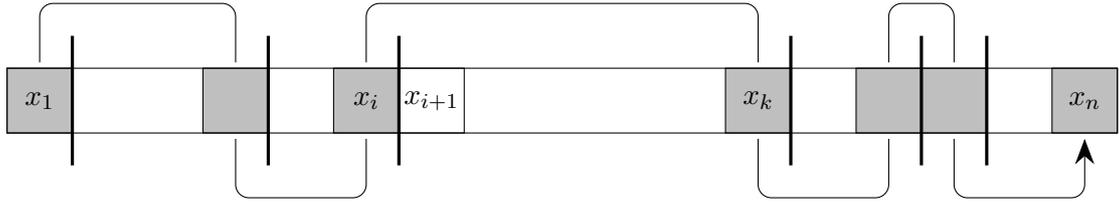

\begin{lemma}\label{last-block-variable-is-unique}
    If $i < n$ and $k=\max\set{t\leq n|T_{i+1} = T_t}$ then there is at most one feasible value for variable $x_k$.
\end{lemma}
\begin{proof}
    If $k < n-1$ then it holds by the harmonic property and the maximality of $k$ that $T_k \geq 2T_{k+1} \geq 2T_{k+2} \geq \dots \geq 2T_{n-1}$ and thus $T_t/T_k \leq 1/2$ for all $t=k+1,\dots,n-1$. Hence,
    \[\frac{c_i+c_k}{T_k} = \sum_{t=i+1}^{n-1} U_t\frac{T_t}{T_k} + \sum_{t=k+1}^{n-1} U_t\frac{T_t}{T_k} = \sum_{t=i+1}^k U_t\underbrace{\frac{T_t}{T_k}}_{=1} + 2\sum_{t=k+1}^{n-1} U_t\underbrace{\frac{T_t}{T_k}}_{\leq 1/2} \leq \sum_{t=i+1}^{n-1} U_t < 1.\]
    If otherwise $k \geq n-1$ then $c_k = 0$ and hence
    \[\frac{c_i+c_k}{T_k} = \frac{c_i}{T_k} = \sum_{t=i+1}^{n-1}U_t\underbrace{\frac{T_t}{T_k}}_{=1} = \sum_{t=i+1}^{n-1} U_t < 1.\]
    By \Cref{uniqueness} this proves the claim.
\end{proof}

This gives rise to the following algorithm. By iterating \Cref{last-block-variable-is-unique} and starting with $x_1 = 1$ we can reveal the last variable of each block of indices of equal periods (cf. \Cref{fig:revealing-flow}). Finally, this reveals the variable $x_n$ and we only need to assure that the value of $x_n$ admits feasible values for variables which are not revealed so far.
Apparently we may restate the constraints of \eqref{response-system} as
\[\ceil*{\frac{J_n-J_j-c_j+T_nx_n}{T_j}} \leq x_j \leq \floor*{\frac{J_n-J_j+T_nx_n}{T_j}} \qquad \forall j=1,\dots,n-1.\]
Therefore, we can simply compare these bounds to assure the existence of a feasible value for each variable $x_j$. See \Cref{alg:reveal} for a formal description.

\begin{algorithm}
    \begin{algorithmic}
        \Procedure{Reveal}{}
            \State $x_1 \gets 1$
            \State $k \gets 1$
            \While{$k < n$}
                \State $i \gets k$
                \State $k \gets \max\set{t \leq n| T_{i+1} = T_t}$
                \If{$\ell_k^{(i)}(x_i) \neq u_k^{(i)}(x_i)$}
                    \State \Return{$-1$}
                \Else
                    \State $x_k \gets \ell_k^{(i)}(x_i)$ \Comment{\Cref{last-block-variable-is-unique}}
                \EndIf
            \EndWhile
            \For{$j=1,\dots,n-1$}
                \If{$\ceil*{\frac{J_n-J_j-c_j+T_nx_n}{T_j}} > \floor*{\frac{J_n-J_j+T_nx_n}{T_j}}$} \Comment{no feasible solution for $x_j$}
                   \State \Return{$-1$}
                \EndIf
            \EndFor
            \State \Return{$x_n$}
        \EndProcedure
    \end{algorithmic}
    \caption{Variable revealing flow}
    \label{alg:reveal}
\end{algorithm}

\begin{observation}
    In fact, by a more sophisticated investigation the number of index blocks of equal periods can be bounded by a constant and thus, the \textup{\textbf{while}} loop reveals $x_n$ in constant time. Therefore, the final feasibility test appears to be the only computational bottleneck.
\end{observation}



\bibliography{ref}
	
\appendix\renewcommand{\thesection}{\Alph{section}}

\section{Hardness of BMS}

We reduce from the problem of \textsc{Directed Diophantine Approximation} with \emph{rounding down}. For any vector $v \in \mathbb{R}^n$ let $\floor{v}$ denote the vector where each component is rounded down, i.e. $(\floor{v})_i = \floor{v_i}$ for all $i\leq n$.

\begin{rectangle}
	\textsc{Directed Diophantine Approximation} with rounding down ($\op{DDA}^{\downarrow}$)\\
	Given: $\alpha_1,\dots,\alpha_n \in \mathbb{Q}_+$, $N \in \integers_{\geq 1}$, $\eps \in \mathbb{Q}$, $0 < \eps < 1$\\
	Decide whether there is a $Q \in \set{1,\dots,N}$ such that
	$\norm{Q\alpha - \floor{Q\alpha}}_{\infty} \leq \eps$.
\end{rectangle}

Eisenbrand and Rothvoß proved that $\op{DDA}^{\downarrow}$ is NP-hard \cite{DBLP:conf/approx/EisenbrandR09}. In fact, every instance of $\op{DDA}^{\downarrow}$ can be expressed as a BMS instance, which yields the following theorem.

\begin{theorem}\label{blaaaa}
	$\op{BMS}$ is \NP-hard (even if $b_i = 0$ for all $i$ with $a_i \neq 0$).
\end{theorem}
\begin{proof}
    Write $\alpha_i = \beta_i/\gamma_i$ for integers $\beta_i\geq 0,\gamma_i\geq 1$ and set $\lambda = \prod_j \beta_j$. Then $\lambda/\alpha_i = (\lambda/\beta_i) \gamma_i \geq 0$ is integer. Let $\mathcal{M}$ denote the following instance of BMS:
    \begin{align}
    	\qquad\qquad 0 \;\; &\leq Q' - (\lambda/\alpha_i) \cdot  y_i &&\leq && \floor{(\lambda / \alpha_i) \cdot \varepsilon}\qquad && \forall i = 1,\dots,n\qquad\;\;\label{hardness-BMS-first}\\
    	\lambda \;\; &\leq Q' - 0\cdot y_{n+1} &&\leq && \lambda \cdot N\label{hardness-BMS-second}\\
    	0 \;\; &\leq Q' - \lambda \cdot y_{n+2} &&\leq && 0\label{hardness-BMS-third}\\
    	&&&&&Q',y_i \in \mathbb{Z} && \forall i=1,\dots,n+2\notag
    \end{align}
    
    So let $Q\in \set{1,\dots,N}$ with $\norm{Q\alpha - \floor{Q\alpha}}_{\infty} \leq \varepsilon$ be given. We obtain readily that $Q' = \lambda Q$ and $y = (\floor{Q\alpha_1},\dots,\floor{Q\alpha_n},0,Q)$ defines a solution of $\mathcal{M}$.
    		
    \noindent Vice-versa let $(Q',y)$ be a solution to $\mathcal{M}$. We see that \eqref{hardness-BMS-first} implies that
    \[0 \leq Q' - (\lambda/\alpha_i) \cdot  y_i \leq  \floor{(\lambda/\alpha_i) \cdot \varepsilon} \leq (\lambda/\alpha_i) \cdot \varepsilon\]
    and by \eqref{hardness-BMS-third} we get $Q' = \lambda \cdot y_{n+2}$ which then implies $0 \leq y_{n+2}\alpha_i - y_i \leq \varepsilon < 1$ for all $i\leq n$. Now, since $y_i$ is integer, there can be only one value for $y_i$, i.e. $y_i = \floor{y_{n+2}\alpha_i}$. By $Q' = \lambda \cdot y_{n+2}$ and \eqref{hardness-BMS-second} we get $y_{n+2} \in \{1,\dots,N\}$ and by setting $Q = y_{n+2}$ this yields $\norm{Q\alpha - \floor{Q\alpha}}_{\infty} \leq \varepsilon$ and that proves the claim.
\end{proof}

\section{Smallest Feasible \texorpdfstring{\textbf{$s$}}{\emph{s}} for Mixing Set}
\begin{lemma}\label{mixing-set-minimize-s-is-simple}
    For $f(s,x) = s$ the mixing set \eqref{mixing-set} can be solved in linear time.
\end{lemma}
\begin{proof}
	We show that $s_{\min} = s^* := \max(\set{0} \cup \set{b_i|a_i = 0})$ where $s_{\min}$ denotes the optimal solution to \eqref{mixing-set} for $f(s,x) = s$. Let $i \leq n$.

	\noindent\textbf{Case} $s^* \geq b_i$. Set $x_i^* = 0$. Then we have $s^* + a_i x_i^* = s^* \geq b_i$.

	\noindent\textbf{Case} $s^* < b_i$.
	Then $a_i \neq 0$ and $b_i - s^* > 0$.
	We set $x_i^* = \ceil{\tfrac1{a_i}(b_i-s^*)}$ if $a_i > 0$ and $x_i^* =
	\floor{\tfrac1{a_i}(b_i-s^*)}$ if $a_i < 0$.
    Again we get that $s^* + a_i x_i^* \geq b_i$.

	Hence, $s^*$ is a solution. Apparently $s^*$ is optimal if $s^* = 0$. If $s^* > 0$ then there is a constraint $j$ with $a_j = 0$ such that $s^* = b_j \leq s_{\min}+a_jx_j = s_{\min}$ for any $x_j$.
\end{proof}


\end{document}